\documentclass[11pt,a4paper]{article}
\pdfoutput=1
\usepackage[left=2.3cm,right=2.3cm,top=2.3cm,bottom=2.3cm]{geometry}
\usepackage[ascii]{inputenc}
\usepackage[dvipsnames,svgnames]{xcolor}
\usepackage[bookmarksnumbered,linktocpage,hypertexnames=false,colorlinks=true,linkcolor=blue,urlcolor=blue,citecolor=blue,anchorcolor=green,breaklinks=true,pdfusetitle]{hyperref}
\usepackage{bbm,braket,microtype,mathrsfs,amsmath,amssymb,xcolor,amsthm,amsfonts,latexsym,mathtools,graphicx,enumitem,booktabs,bm,xspace,float,mathdots,caption,subcaption,ellipsis,mleftright}
\usepackage[T1]{fontenc}
\usepackage{textcomp}
\usepackage[english]{babel}
\usepackage[capitalize,nameinlink]{cleveref}
\usepackage{mathpazo}
\usepackage[normalem]{ulem}
\newtheorem{thm}{Theorem}[section]\crefname{thm}{Theorem}{Theorems}
\newtheorem{lem}[thm]{Lemma}

\newtheorem{cor}[thm]{Corollary}

\theoremstyle{definition}

\AddToHook{env/lem/begin}{\crefalias{thm}{lem}}
\AddToHook{env/prp/begin}{\crefalias{thm}{prp}}
\AddToHook{env/cor/begin}{\crefalias{thm}{cor}}
\AddToHook{env/cnj/begin}{\crefalias{thm}{cnj}}
\AddToHook{env/dfn/begin}{\crefalias{thm}{dfn}}
\AddToHook{env/rem/begin}{\crefalias{thm}{rem}}
\AddToHook{env/exa/begin}{\crefalias{thm}{exa}}
\numberwithin{equation}{section}
\newcommand{\Ex}{\mathop{\bf E\/}}
\DeclareMathOperator{\tr}{tr}
\DeclareMathOperator{\rk}{rk}

\DeclareMathOperator{\Sym}{Sym}

\DeclareMathOperator{\Lin}{L}

\DeclareMathOperator{\U}{U}
\DeclareMathOperator{\GL}{GL}
\DeclareMathOperator{\Sp}{Sp}
\DeclareMathOperator{\Mp}{Mp}
\DeclareMathOperator{\ASp}{ASp}

\DeclareMathOperator{\SO}{SO}
\DeclareMathOperator{\Spin}{Spin}

\DeclareMathOperator{\PSD}{P}
\renewcommand{\O}{\operatorname{O}}
\renewcommand{\S}{\operatorname{S}}
\renewcommand{\L}{\operatorname{L}}
\DeclarePairedDelimiter\abs{\lvert}{\rvert}

\DeclarePairedDelimiter\parens{\lparen}{\rparen}

\allowdisplaybreaks[4]
\newcommand{\R}{\mathbbm R}
\newcommand{\RR}{\mathbbm R}
\newcommand{\C}{\mathbbm C}

\newcommand{\E}{\mathcal E}
\newcommand{\cA}{\mathcal A}
\newcommand{\cB}{\mathcal B}

\newcommand{\cP}{\mathcal P}
\newcommand{\cH}{\mathcal H}
\newcommand{\cK}{\mathcal K}
\newcommand{\cL}{\mathcal L}
\newcommand{\cR}{\mathcal R}
\newcommand{\cV}{\mathcal V}
\newcommand{\cW}{\mathcal W}
\newcommand{\cG}{\mathcal G}

\newcommand{\so}{\mathfrak{so}}
\newcommand{\ot}{\otimes}

\newcommand{\eps}{\varepsilon}
\newcommand{\bit}{\{0,1\}}
\newcommand{\id}{I}
\newcommand{\idCh}{\mathcal I}
\newcommand*{\tran}{^{\mkern-1.5mu\mathsf{T}}}

\newcommand{\dd}{\mathrm{d}}
\newcommand{\bigO}{\mathcal O}
\newcommand{\ketbra}[2]{\mathinner{\lvert#1\rangle\!\langle#2\rvert}}
\newcommand{\proj}[1]{\mathinner{\lvert#1\rangle\!\langle#1\rvert}}

\begin{document}

\title{A random purification channel for arbitrary symmetries with~applications to fermions and bosons}
\author{Michael Walter\thanks{LMU Munich and University of Amsterdam \& QuSoft, Amsterdam, \href{mailto:michael.walter@lmu.de}{michael.walter@lmu.de}} \and Freek Witteveen\thanks{CWI \& QuSoft, Amsterdam, \href{mailto:}{f.witteveen@cwi.nl}}}
\date{}
\maketitle
\begin{abstract}
The random purification channel maps $n$ copies of any mixed quantum state~$\rho$ to $n$ copies of a random purification of~$\rho$.
We generalize this construction to arbitrary symmetries:
for any group~$G$ of unitaries, we construct a quantum channel that maps states contained in the algebra generated by~$G$ to random purifications obtained by twirling over~$G$.
In addition to giving a surprisingly concise proof of the original random purification theorem, our result implies the existence of fermionic and bosonic Gaussian purification~channels.
As applications, we obtain the first tomography protocol for fermionic Gaussian states that scales optimally with the~number~of~modes and the error, as well as an improved property test for this class of states.
\end{abstract}

\section{Introduction}
The church of the larger Hilbert space teaches that one can interpret any mixed state as the reduced state of a pure state, and that similarly any process can be purified by the Stinespring extension.
Indeed, for any mixed state~$\rho$ on a quantum system with Hilbert space $\cH$, there exists a pure state~$\ket{\psi_\rho}$ on $\cH \ot \cH'$, where the reference system $\cH'$ is a copy of $\cH$, such that the reduced state on~$\cH$ equals~$\rho$.
This purification is not unique: applying an arbitrary unitary~$U \in \U(\cH')$ to the reference system yields another purification, and in fact any two purifications of~$\rho$ differ by such a unitary.
This is an extremely useful principle in quantum information theory, where it is often much easier to reason about pure states than about mixed states.

How does one put this creed into practice?
Recently, it has been shown that purifications not only exist, but can in fact be constructed in a meaningful way by a \emph{quantum channel} \cite{tang2025conjugate,pelecanos2025mixed,girardi2025random}, see also \cite{soleimanifar2022testing,chen2024local}.
While no quantum channel can assign a fixed purification to every quantum state, there does exist a \emph{random purification channel} that prepares a purification with a random unitary applied to the reference system.
Crucially, it can be extended to multiple copies of the state, such that it does so consistently.
That is, the random purification channel $\cP_{n,d}$ acts on $n$ copies of a quantum system with Hilbert space $\cH = \C^d$, and it has the property that for every state $\sigma$ on $\cH$,
\begin{align}\label{eq:random purification intro}
    \cP_{n,d}[\sigma^{\ot n}] = \int_{\U(d)} \left((I \ot U) \psi^\text{std}_\sigma (I \ot U^\dagger)\right)^{\ot n} \, \dd U = \Ex_{\psi_\sigma \text{ purification of } \sigma } \psi_\sigma^{\ot n},
\end{align}
where $\psi^\text{std}_\sigma$ is a fixed (standard) purification of $\sigma$.%

This construction has useful applications.
In particular, it has been applied in \cite{pelecanos2025mixed} to improve tomography for mixed states.
The idea is that if one has a tomography protocol for pure states, then this automatically extends to a tomography protocol for mixed states, as follows:
\begin{enumerate}[noitemsep]
    \item Given $n$ copies of a mixed state~$\sigma$, apply the random purification channel.
    \item Apply pure state tomography on the output of the random purification channel.
    \item Return the reduced state of the estimated pure state.
\end{enumerate}
To learn the mixed state, it clearly suffices to learn a purification of it, and it doesn't matter which one.
Therefore, the above sketched protocol is correct.
Note that for this application, it is crucial that the random purification is consistent over multiple copies.
The above protocol, and more sophisticated versions of it, has been analyzed in detail in~\cite{pelecanos2025mixed}.
Remarkably,while the analysis is \emph{easier} than for previous protocols, the resulting protocol is \emph{optimal} in terms of sample complexity.
Additionally, the random purification channel can be implemented efficiently, and this yielded the first tomography scheme that is both optimal in sample complexity and computationally efficient.
Another application is to channel tomography, where one can apply random purification to the Choi state of a quantum channel, with an improved sample complexity compared to previous channel tomography results~\cite{mele2025optimal,chen2025quantum}.

The central insight behind the random purification channel is drawn from the church of the symmetric subspace~\cite{harrow2013church}, which teaches that in quantum information tasks involving independent copies of a state, both the design and analysis of optimal protocols are often governed by symmetries and representation theory.
To see the connection, we observe that states of the form~$\ket{\psi}^{\ot n}$ span the symmetric subspace in $\cH^{\ot n}$.
More generally, on~$\cH^{\ot n}$ there are two group actions. One is by $\U(d)$, with unitaries acting identically on each copy (i.e., by $U^{\ot n}$). The other action is by the permutation group~$S_n$, which permutes the copies. These two actions are each others commutants, and this gives the Schur-Weyl decomposition, which decomposes the actions of $\U(d)$ and $S_n$ to a block-diagonal form:
\begin{align}\label{eq:sw intro}
    \cH^{\ot n} \cong \bigoplus_{\lambda} \mathcal V_\lambda \ot \mathcal W_\lambda,
\end{align}
where $\mathcal V_\lambda$ and $\mathcal W_\lambda$ are irreducible representations of $\U(d)$ and $S_n$ respectively, labeled by Young diagrams~$\lambda$. The high-level idea of the random purification channel is that $\rho^{\ot n}$ is in the algebra generated by the action of~$\U(d)$, hence
\begin{align}\label{eq:rho n times}
    \rho^{\ot n} \cong \bigoplus_{\lambda} \rho_{\lambda} \ot \frac{\id_{W_{\lambda}}}{\dim \mathcal W_{\lambda}},
\end{align}
so the state is maximally mixed on the irreducible representations of $S_n$.
The random purification channel now takes a reference system with the same decomposition.
For each $\lambda$, it prepares the maximally entangled state on each pair of $\cW_{\lambda}$ (which indeed purifies the maximally mixed state in \cref{eq:rho n times}), and it prepares the maximally mixed state on the reference copies of the $\mathcal V_\lambda$.
Using the fact that averaging over random unitaries~$U^{\ot n}$ corresponds to a depolarizing channel on the~$\mathcal V_\lambda$, one can then show that this procedure coincides with \cref{eq:random purification intro}.
In fact, one does not need the details of Schur-Weyl duality to derive this result, a simple alternative derivation is given in \cite{girardi2025random}, and below we explain how to obtain this result without any use of Schur-Weyl duality.
The advantage of using the explicit decomposition in \cref{eq:sw intro} is that it is efficiently implemented by the quantum Schur transform---this gives a computationally efficient implementation of the random purification channel.

In this work, we introduce a natural generalization of the random purification channel for arbitrary symmetry groups or algebras.
While our result is more general, it also gives a simple and transparent proof of the standard random purification channel.
As an application we consider fermionic Gaussian states on $m$ modes, which are generated by unitaries of the form~$U_R^{\ot n}$, where the $U_R$ are a projective representation of $R \in \SO(2m)$.
The generalized random purification channel can then easily be applied to reduce the case of mixed fermionic Gaussian states to that of pure fermionic Gaussian states, which allows us to obtain new protocols for tomography and property testing.
We now describe our contributions in more detail.

\subsection{Random purification channel for general symmetries}
Let $\cH$, $\cH'$ be Hilbert spaces of the same dimension, with bases~$\ket x$,~$\ket{x'}$.
For an operator~$A$ on the first Hilbert space~$\cH$, denote by~$A\tran$ its transpose on the other Hilbert space~$\cH'$, where we identify the operators through the choices of basis.
Also denote by~$\ket\Gamma = \sum_x \ket{xx'} \in \cH \ot \cH'$ the corresponding unnormalized maximally entangled state.
For any quantum state $\rho \in \S(\cH)$, we have a purification $\ket{\psi^\text{std}_\rho} \coloneqq (\sqrt\rho \ot I) \ket\Gamma$, called the \emph{standard purification}.%
\footnote{To obtain a truly canonical purification one can choose $\cH'$ to be the dual (or the conjugate) Hilbert space of~$\cH$.}

\begin{thm}[Random purification for general symmetries, simplified]\label{thm:main simplified}
For any closed subgroup~$G \subseteq \U(\cH)$, there is a channel $\cP_G \colon \Lin(\cH) \to \Lin(\cH \ot \cH')$ such that the following holds:
For any state~$\rho \in \S(\cH)$ such that~$\rho \in \C G$ (the span of~$G$), we have
\begin{align}\label{eq:twirl intro}
    \cP_G[\rho] = \int_G (I \ot g\tran) \psi^\text{std}_\rho (I \ot \bar g) \, \dd g.
\end{align}
If the Fourier transform for the action of~$G$ can be implemented efficiently, then so can the channel.
\end{thm}

\noindent
\Cref{thm:main simplified} is a simplified version of \cref{thm:main technical}, which gives a precise description of the channel and is naturally phrased in the language of *-algebras.%
\footnote{Note that $\C G = G''$ is automatically closed under products and adjoints, hence a *-algebra.}

It is easy to see that the the random purification channel of~\cite{tang2025conjugate,pelecanos2025mixed,girardi2025random} is a special case of the above theorem.
Let $\cH = \cH' = (\C^d)^{\ot n}$, with the computational basis. 
If~$\rho = \sigma^{\ot n}$ is an IID state, then~$\rho \in \C \{ U^{\ot n} : U \in \U(d) \}$.%
\footnote{\label{footnote:rho}This follows from Schur-Weyl duality, but also from the following easy argument (see, e.g., \cite{symqi}). As a vector space, $\C \{ U^{\ot n} : U \in \U(d) \}$ contains arbitrary limits, hence in particular the Lie algebra action, therefore also its complexification, and thus also $g^{\ot n}$ for $g \in \overline{\GL(d)} = \C^{d \times d}$.}
Because the unitary group is closed under transposes, the standard purification of an $n$-th tensor power is the $n$-th tensor power of the standard purification, and purifications are unique up to unitaries, we obtain the following as a direct consequence of \cref{thm:main simplified}:

\begin{cor}[\cite{tang2025conjugate,pelecanos2025mixed,girardi2025random}]\label{cor:og channel}
There is a channel~$\cP^\text{Purify}_{n,d}$ from~$(\C^d)^{\ot n}$ to the symmetric subspace of~$(\C^d \ot \C^d)^{\ot n}$ such~that
\begin{align*}
    \cP^\text{Purify}_{n,d}[\sigma^{\ot n}]
= \int_{\U(d)} (I \ot U^{\ot n}) \psi^{\text{std}}_{\sigma^{\ot n}} (I \ot U^{\dagger,\ot n}) \ \dd U
= \Ex_{\psi_\sigma \text{ purification of } \sigma } \psi_\sigma^{\ot n}
\end{align*}
The first equation holds not just IID states~$\sigma^{\ot n}$, but for arbitrary permutation-invariant states~$\rho$.
\end{cor}

This is due to Schur-Weyl duality, which asserts that $\C \{ U^{\ot n} : U \in \U(d) \}$ is the commutant of the permutation operators.
Moreover, the channel can be implemented efficiently by using the quantum Schur transform~\cite{bacon2006efficient,burchardt2025high}.

\subsection{Application: Sample-optimal tomography and testing of fermionic Gaussian states}
One application of our random purification theorem is in the setting of fermionic Gaussian states.
These are states that live on exponentially large Hilbert space, the fermionic Fock space~$\bigwedge \C^m \cong (\C^2)^{\ot m}$, yet can be described succinctly (e.g., in terms of a $2m \times 2m$ covariance matrix).
There is an associated group of Gaussian unitaries, which form a projective representation of~$\SO(2m)$.
We refer to \cref{sec:fermions} for more details.
Similarly as in the case of qudits, for any fermionic Gaussian state~$\sigma$, $\sigma^{\ot n}$ is contained in the span of the $n$-th tensor powers of Gaussian unitaries.
This allows us to apply our results.
By specializing \cref{thm:main simplified} to this setting, we obtain a quantum channel that sends $n$ copies of a fermionic Gaussian state to a random fermionic Gaussian purification:%
\footnote{Just like \cref{cor:og channel}, which applies to arbitrary permutation-invariant states, \cref{cor:fermions} generalizes to states invariant under an action of the orthogonal group~$\O(n)$, as a consequence of Howe duality~\cite{howe1995perspectives,girardi2025gaussian}.}

\begin{cor}[Fermionic Gaussian random purification]\label{cor:fermions}
For any $m,n$, there is a channel $\cP_{n,m}^\text{Fermi}$ from $(\bigwedge \C^m)^{\ot n}$ to $(\bigwedge \C^{2m})^{\ot n}$ such that the following holds:
for any fermionic Gaussian state~$\sigma$ on $\bigwedge \C^m$,
\begin{align*}
    \cP_{n,m}^\text{Fermi}[\sigma^{\ot n}]
= \Ex_{\psi_\sigma \text{ Gaussian purification of } \sigma } \psi_\sigma^{\ot n}
= \int_{\SO(2m)} (I \ot U_R^{\ot n}) \psi^{\text{std},\ot n}_\sigma (I \ot U_R^{\dagger,\ot n}) \ \dd R,
\end{align*}
where the expectation is over Gaussian purifications obtained by applying a uniformly random Gaussian unitary~$U_R$ to the standard purification, as in the right-hand side formula.
\end{cor}

We use \cref{cor:fermions} to construct a tomography protocol for fermionic Gaussian states.
Tomography of fermionic Gaussian states has been studied in previous work \cite{aaronson2021efficient,ogorman2022fermionic,bittel2025optimalfermion}. The previous best sample complexities were $\widetilde \bigO(m^2)$ for pure states\footnote{Since \cite{zhao2024learning} shows a sample complexity linear in $L$ for states prepared by $L$ two-qubit gates and all Gaussian fermionic states can be prepared using $\bigO(m^2)$ two-qubit gates.} \cite{zhao2024learning} and $\bigO(m^4)$ for mixed states \cite{bittel2025optimalfermion}.
Our work improves this to~$\bigO(m^2)$ for both pure and mixed states:

\begin{thm}[Fermionic Gaussian tomography, upper bound]\label{thm:tomo}
There exists a tomography protocol for $m$-mode fermionic Gaussian states~$\sigma$ that outputs an estimate~$\hat\sigma$ with fidelty at least~$F(\sigma,\hat \sigma)^2 \geq 1 - \eps$ with constant probability, using $\bigO(m^2/\eps)$ copies of~$\sigma$.
\end{thm}

Our protocol first applies the random purification channel~$\cP_{n,d}^\text{Fermi}$, and then performs a natural tomography protocol for pure \emph{Gaussian} states analogous to the well-known uniform POVM for pure-state tomography~\cite{hayashi1998asymptotic}, following the template of \cite{pelecanos2025mixed}.
We also show that this sample complexity is optimal, and we show this is already the case for pure states.

\begin{thm}[Fermionic Gaussian tomography, lower bound]\label{thm:tomo lower bound}
Suppose there exists a tomography protocol for $m$-mode pure fermionic Gaussian states~$\psi$ that outputs an estimate~$\hat\psi$ with fidelty at least~$F(\psi,\hat \psi)^2 \geq 1 - \eps$ with probability at least $0.99$ for any pure Gaussian state. Then this protocol needs to use at least $\Omega(m^2/\eps)$ copies of~$\psi$.
\end{thm}

Our improved tomography algorithm also gives an improvement in the sample complexity of testing whether a pure state is a fermionic Gaussian state, or bounded away from the set of these states.

\begin{cor}[Property testing fermionic Gaussian pure states]\label{cor:testing fermions}
    The sample complexity of pure Gaussianity testing on $m$ fermionic modes is at most~$n = \bigO(m^2 / \eps)$.
    That is, there is an algorithm that, using $n$ copies of a pure state~$\psi$ can distinguish with constant probability of success whether $\psi$ is Gaussian or $\abs{\braket{\psi|\phi}}^2 \leq 1 - \eps$ for any Gaussian state $\ket{\psi}$ (assuming one of the two is the case).
\end{cor}

\subsection{Random purification of bosonic Gaussian states}
The random purification theorem also applies to bosonic Gaussian states that are gauge-invariant.
Gaussian states are quantum states on the infinite-dimensional Hilbert space of $m$ bosonic modes, $\cH = \L^2(\R^m)$, that are succinctly described, e.g.\ by a $2m\times2m$ covariance matrix and a mean vector.
There is also a natural group of Gaussian unitaries, which form a projective representation of the affine symplectic group~$\ASp(2m)$.
See \cref{sec:bosons} for more details.
One challenge is that the symplectic group (affine or not) is non-compact.
This means that there is no invariant probability measure---in particular, the notion of a uniformly random Gaussian unitary and twirl as in \cref{eq:twirl intro} is not well-defined.
We will return to this point later.

However, we can also consider the maximally compact subgroup~$\Sp(2m) \cap \O(2m) \cong \U(m)$, which corresponds to the so-called \emph{passive} Gaussian unitaries and has an invariant probability measure.
The corresponding quantum states are known as the \emph{gauge-invariant} (or number-preserving) Gaussian states.
Similarly as before, if $\sigma$ is a gauge-invariant Gaussian state, then $\sigma^{\ot n}$ is in the closed span of the $n$-th tensor powers of passive Gaussian unitaries.
This yields the~following~result.

\begin{cor}[Bosonic gauge-invariant Gaussian random purification]\label{cor:bosons}
For any $m,n$, there is a channel $\cP_{n,m}^\text{Boson}$ from $\L^2(\R^m)^{\ot n}$ to $\L^2(\R^{2m})^{\ot n}$ such that the following holds:
for any gauge-invariant bosonoic Gaussian state~$\sigma$ on $\L^2(\R^m)$,
\begin{align*}
    \cP_{n,m}^\text{Boson}[\sigma^{\ot n}] =
\Ex_{\substack{\psi_\sigma \text{ Gaussian purification of } }\sigma} \psi_\sigma^{\ot n} =
\int_{\U(m)} (I \ot U_O^{\ot n}) \psi^{\text{std},\ot n}_\sigma (I \ot U_O^{\dagger,\ot n}) \ \dd O,
\end{align*}
where the expectation is over Gaussian purifications obtained by applying a uniformly random passive Gaussian unitary~$U_O$ to the standard purification, as in the right-hand side formula.
\end{cor}

This could be a useful tool for designing improved tomography protocols for gauge-invariant Gaussian states, exploring which we leave to future work.

We now return to the case of general Gaussian states and unitaries.
There exists a decomposition similar to the Schur-Weyl decomposition in \cref{eq:sw intro}, due to Kashiwara and Vergne~\cite{kashiwara1978segal,girardi2025gaussian}, where the symplectic group---or its double cover, the metaplectic group---acts by $n$-th tensor powers of Gaussian unitaries, and the commutant is spanned by the action of the orthogonal group~$\O(n)$ on $\cH^{\ot n} \cong \L^2(\R^{m \times n})$.%
\footnote{ Similarly, the affine symplectic or metaplectic group are dual to the stochastic orthogonal group.}
This gives a decomposition as in \cref{eq:sw intro}, where the~$\cV_{\lambda}$ and~$\cW_{\lambda}$ are irreducible representations of the metaplectic and orthogonal groups, respectively.
From this perspective, the challenge of constructing a uniformly random purification channel is closely related to the fact that the irreducible representations of the metaplectic group are infinite-dimensional, and therefore do not admit a maximally mixed state.
However, for applications such as tomography, it is not necessary that the purifications are \emph{uniformly} random.
We therefore view the construction of random purification channels for general bosonic Gaussian states---without assuming gauge invariance and possibly under suitable energy constraints and with a controlled approximation error---as an interesting open problem for future work.

\paragraph*{Related work.}
Closely related work has been done independently by Mele, Girardi, Chen, Fanizza, and Lami~\cite{relatedwork}.
We have arranged with these authors to coordinate the submission of our papers.

\section{Preliminaries and Notation}\label{sec:prelim}

Unless stated otherwise, all Hilbert spaces are finite-dimensional.
We write~$\Lin(\cH)$ for the linear operators on a Hilbert space~$\cH$, $\PSD(\cH)$ for the positive semidefinite operators on~$\cH$, and~$\S(\cH)$ for the set of density matrices on~$\cH$.
A quantum channel is a completely positive trace-preserving map~$\Phi\colon\Lin(\cH)\to\Lin(\cK)$, where~$\cH$ and~$\cK$ are Hilbert spaces.
We recall that any closed subgroup~$G \subseteq \U(\cH)$ has a Haar probability measure, which we denote by~$\dd g$.

\subsection{Algebras}
We consider a $*$-algebra of linear operators~$\cA \subseteq \Lin(\cH)$.
Any such algebra can be written as a direct sum of matrix algebras.
Concretely, there exists a decomposition of the Hilbert space
\begin{align}\label{eq:decomposition}
    \cH \cong \bigoplus_{\lambda \in \Lambda} \cL_\lambda \ot \cR_\lambda,
\end{align}
where~$\Lambda$ is a finite index set and $\cL_\lambda$ and $\cR_\lambda$ are Hilbert spaces, such that
\begin{align}\label{eq:decomposition algebras}
    \cA \cong \bigoplus_{\lambda \in \Lambda} \Lin(\cL_\lambda) \ot I_{\cR_\lambda},
    \quad\text{and hence}\quad
    \cA' \cong \bigoplus_{\lambda \in \Lambda} I_{\cL_\lambda} \ot \Lin(\cR_\lambda),
\end{align}
where $S' \subseteq \Lin(\cH)$ denotes the commutant of a subset~$S = S^* \subseteq \Lin(\cH)$, that is, the $*$-algebra of operators that commute with all the operators in~$S$.

We let~$\E_\cA \colon \Lin(\cH) \to \cA$ denote the orthogonal projection onto~$\cA$ with respect to the Hilbert-Schmidt inner product.
With respect to~\eqref{eq:decomposition}, it is~given~by
\begin{align}\label{eq:E}
    \E_\cA[M] \cong \bigoplus_{\lambda \in \Lambda} \tr_{\cR_\lambda}[ P_\lambda M P_\lambda ] \ot \frac{I_{\cR_\lambda}}{\dim \cR_\lambda}
    \qquad (M \in \Lin(\cH)),
\end{align}
where $P_\lambda$ denotes the orthogonal projections onto the summands in \cref{eq:decomposition}.
That is, $\E_\cA$ acts as a measurement of~$\lambda$ composed with completely depolarizing channels on each~$\cR_\lambda$.
Alternatively, the orthogonal projection can be written as follows:
Take any closed subgroup~$G \subseteq \U(\cH)$ that generates (equivalently, spans) the commutant algebra~$\cA'$.
Equivalently, by the double commutant theorem, $G' = \cA$.
Then the orthogonal projection is given by the Haar twirl over~$G$:
\begin{align}\label{eq:twirl}
    \E_\cA[M] = \int_G g M g^\dagger \,\dd g.
\end{align}
A choice that always works is to take~$G = \U(\cA')$, the group of unitaries in the $*$-algebra, but often~$G$ can be taken to be much smaller.%
\footnote{For example, the Pauli group spans the $*$-algebra of $2\times 2$-matrices.}
More generally, we may consider any ensemble of unitaries satisfying~\eqref{eq:twirl} (a ``unitary 1-design'' for~$\cA'$).

A well-known and motivating example is the following: if~$\cH = \cK^{\ot t}$ and $\cA \subseteq \Lin(\cK^{\ot t})$ is spanned by the permutation operators, then $\cA'$ is spanned by the group~$G = \{ U^{\ot t} : U \in \U(\cK) \}$.
This is known as \emph{Schur-Weyl duality}.

\section{A random purification channel for general symmetries}
In this section we state and prove a random purification theorem for general symmetries.
Throughout this section, let $\cH \cong \cH'$ be Hilbert spaces of the same dimension, with fixed bases~$\ket x$,~$\ket{x'}$.
For an operator~$A \in \Lin(\cH)$, we identify an operator~$A' \in \Lin(\cH')$ through the basis choice, so $\braket{x' | A | y'} = \braket{x | A | y}$. In a slight abuse of notation, when it is clear on which Hilbert space the operator acts, we write $A'=A$.
The transpose and conjugation are defined with respect to the choice of basis on $\cH$ or $\cH'$.
If we let~$\ket\Gamma = \sum_x \ket{xx'} \in \cH \ot \cH'$ denote the unnormalized maximally entangled state corresponding to these bases, then we have the following version of the transpose trick: $(A \ot I) \ket\Gamma = (I \ot A\tran) \ket\Gamma$.
For any quantum state $\rho \in \S(\cH)$, we define the standard purification by the formula: $\ket{\psi^\text{std}_\rho} \coloneqq (\sqrt\rho \ot I) \ket\Gamma$.
We denote by $(\cH \ot \cH')^{\U(\cA)}$ the subspace which consists of the vectors that are invariant under~$u \ot \bar u$ (not $u \ot \bar u$!) for any unitary~$u \in \U(\cA)$ in the algebra~$\cA$, and call this the \emph{$\cA$-symmetric subspace}.

We now present our general random purification theorem.
We note that the algebra~$\cA$ plays the role of the commutant of the group~$G$ in \cref{thm:main simplified}, which is an immediate corollary.

\begin{thm}[Random purification for general symmetries]\label{thm:main technical}
For any $*$-algebra $\cA \subseteq \Lin(\cH)$, there is a quantum channel
$\cP_\cA \colon \Lin(\cH) \to \Lin(\cH \ot \cH')$
with the following properties:
\begin{enumerate}
\item \emph{Random purification of invariant input states:}
For any input state~$\rho \in \S(\cH)$ that commutes with~$\cA$, i.e., $\rho \in \cA'$, we have
\begin{align}\label{eq:sym action}
    \cP_\cA[\rho]
= \parens*{ \idCh \ot \E_{\cA\tran} }[ \psi^\text{std}_\rho ]
= \int_G (I \ot g\tran) \psi^\text{std}_\rho (I \ot \bar g) \,\dd g,
\end{align}
where $G \subseteq \U(\cH)$ is any closed subgroup with~$\C G = \cA'$.

\item \emph{Symmetry of output states:}
For any input state~$\rho \in \S(\cH)$, the output state $\cP_\cA[\rho]$ is supported on the $\cA$-symmetric subspace~$(\cH \ot \cH')^{\U(\cA)}$ 
and it is also $G$-invariant on the purifying system (commutes with~$I \ot g\tran$ for all $g \in G$).

\item \emph{Efficiency:}
If unitaries implementing the decomposition~\eqref{eq:decomposition} and the basis change~$\ket x \mapsto \ket{x'}$ can be implemented efficiently, then $\cP_\cA$ can be implemented efficiently.


\item \emph{Explicit formulas:}
For all $\rho \in S(\cH)$,
\begin{align}\label{eq:explicit}
    \cP_\cA[\rho]
= \Pi_\cA \parens*{ \rho \ot \Omega } \Pi_\cA
= \sqrt{\cP_\cA[I]} \parens*{ \rho \ot I } \sqrt{\cP_\cA[I]},
\end{align}
where $\Pi_\cA$ denotes the projection onto the~$\cA$-symmetric subspace, and~$\Omega \coloneqq \bigoplus_{\lambda\in\Lambda} \frac {\dim\cL_\lambda}{\dim\cR_\lambda} P_\lambda\tran$, with~$P_\lambda$ the orthogonal projections onto the direct summands in~\eqref{eq:decomposition}.
\end{enumerate}
\end{thm}
\begin{proof}
Recall the decompositions~\eqref{eq:decomposition} and~\eqref{eq:decomposition algebras}, which we may assume hold with equality:
\begin{align*}
    \cH = \bigoplus_{\lambda \in \Lambda} \cL_\lambda \ot \cR_\lambda, \quad
    \cA = \bigoplus_{\lambda \in \Lambda} \Lin(\cL_\lambda) \ot I_{\cR_\lambda},
    \quad
    \cA' = \bigoplus_{\lambda \in \Lambda} I_{\cL_\lambda} \ot \Lin(\cR_\lambda).
\end{align*}
Furthermore, we may assume that we have a similar decomposition
\begin{align*}
    \cH' = \bigoplus_{\lambda \in \Lambda} \cL'_\lambda \ot \cR'_\lambda, \quad
    \cA\tran = \bigoplus_{\lambda \in \Lambda} \Lin(\cL'_\lambda) \ot I_{\cR'_\lambda},
    \quad
    (\cA')\tran = \bigoplus_{\lambda \in \Lambda} I_{\cL'_\lambda} \ot \Lin(\cR'_\lambda),
\end{align*}
where $\cL'_\lambda = \cL_\lambda$ and $\cR'_\lambda = \cR_\lambda$ (we still use primes to distinguis the purifying space from the original one), and that we work with a basis obtained by picking bases on these tensor factors.%
\footnote{\label{footnote:HHprime}To see this, note that if $U$ implements the decomposition for~$\cH$ and $J \colon \cH' \to \cH, \ket{x'} \mapsto \ket{x}$.
Then, $\bar U \circ \bar J$ implements the desired decomposition for~$\cH'$.}
In particular:
\begin{align}\label{eq:transpose}
    x = \bigoplus_{\lambda \in \Lambda} a_\lambda \ot b_\lambda
\quad\Rightarrow\quad
    x\tran = \bigoplus_{\lambda \in \Lambda} a_\lambda\tran \ot b_\lambda\tran.
\end{align}
In particular, if $P_\lambda$ denotes the projections onto the direct summands of~$\cH$, then~$P_\lambda\tran$ are the projections onto the direct summands of~$\cH'$.
\begin{enumerate}
\item We will compute $\parens*{ \idCh \ot \E_{\cA\tran} }[ \psi^\text{std}_\rho ]$ and see that it can be realized by a natural quantum channel acting on~$\rho$.
We start by computing the standard purification with respect to the above decomposition.
Because~$\ket\Gamma$ has the property that~$(U \ot \bar U)\ket\Gamma = \ket\Gamma$ for all~$U \in \U(\cH)$, we see from \cref{eq:transpose} that
\begin{align}\label{eq:EPR}
    \ket\Gamma = \bigoplus_\lambda \ket{\Gamma}_{\cL_\lambda\cL'_\lambda} \ot \ket{\Gamma}_{\cR_\lambda\cR'_\lambda}
\in \cH \ot \cH',
\end{align}
where the right-hand side are the unnormalized maximally entangled states in~$\cL_\lambda \ot \cL'_\lambda$ and~$\cR_\lambda \ot \cR'_\lambda$ (with respect to the bases chosen above), with squared norm equal to $\dim(\cL_\lambda)$ and $\dim(\cR_\lambda)$, respectively.
Since $\rho \in \cA'$, we see that
\begin{align}\label{eq:rho}
    \rho = \bigoplus_{\lambda\in\Lambda} I_{\cL_\lambda} \ot \rho_\lambda,
\quad\text{where}\quad
    \rho_\lambda = \frac1{\dim\cL_\lambda} \tr_{\cL_\lambda}[P_\lambda \rho P_\lambda] \in \PSD(\cR_\lambda).
\end{align}
Thus, the standard purification is given by
\begin{align*}
    \ket{\psi^\text{std}_\rho} = \bigoplus_{\lambda\in\Lambda} \ket{\Gamma}_{\cL_\lambda\cL_\lambda'} \ot \ket{\psi_\lambda}_{\cR_\lambda\cR_\lambda'},
\end{align*}
where each $\ket{\psi_\lambda}$ is the standard purification of~$\rho_\lambda$ (withr respect to the bases chosen above).
As a density operator,
\begin{align*}
    \proj{\psi^\text{std}_\rho} = \bigoplus_{\lambda,\mu\in\Lambda} \ket{\Gamma}_{\cL_\lambda\cL_\lambda'}\bra{\Gamma}_{\cL_\mu\cL_\mu'} \ot \ket{\psi_\lambda}_{\cR_\lambda\cR_\lambda'}\bra{\psi_\mu}_{\cR_\mu\cR_\mu'},
\end{align*}
From \cref{eq:E}, we see that
\begin{align*}
    \parens*{ \idCh \ot \E_{\cA\tran} }\mleft[ \proj{\psi^\text{std}_\rho} \mright]
&= \bigoplus_{\lambda\in\Lambda} \ket{\Gamma}_{\cL_\lambda\cL_\lambda'}\bra{\Gamma}_{\cL_\lambda\cL_\lambda'} \ot \tr_{\cR_\lambda'}[\ket{\psi_\lambda}_{\cR_\lambda\cR_\lambda'}\bra{\psi_\lambda}_{\cR_\lambda\cR_\lambda'}] \ot \frac{I_{\cR'_\lambda}}{\dim \cR_\lambda} \\
&= \bigoplus_{\lambda\in\Lambda} \ket{\Gamma}_{\cL_\lambda\cL_\lambda'}\bra{\Gamma}_{\cL_\lambda\cL_\lambda'} \ot \rho_{\lambda} \ot \frac{I_{\cR'_\lambda}}{\dim \cR_\lambda} \\
&= \bigoplus_{\lambda\in\Lambda} \frac{\ket{\Gamma}_{\cL_\lambda\cL_\lambda'}\bra{\Gamma}_{\cL_\lambda\cL_\lambda'}}{\dim\cL_\lambda} \ot \tr_{\cL_\lambda}[P_\lambda \rho P_\lambda] \ot \frac{I_{\cR'_\lambda}}{\dim \cR_\lambda}.
\end{align*}
We may take the last line as the \emph{definition} of the desired channel: for all $\rho \in \S(\cH)$, we define
\begin{align}\label{eq:def in proof}
    \cP_\cA[\rho] := \bigoplus_{\lambda\in\Lambda} \frac{\ket{\Gamma}_{\cL_\lambda\cL_\lambda'}\bra{\Gamma}_{\cL_\lambda\cL_\lambda'}}{\dim\cL_\lambda} \ot \tr_{\cL_\lambda}[P_\lambda \rho P_\lambda] \ot \frac{I_{\cR'_\lambda}}{\dim \cR_\lambda},
\end{align}
which satisfies the first equation in~\eqref{eq:sym action}; the second equation follows at once from~\eqref{eq:twirl}.

\item
Because $\U(\cA) = \bigoplus_{\lambda\in\Lambda} \U(\cL_\lambda) \ot I_{\cR_\lambda}$, it is clear from \cref{eq:transpose} that
the projection onto the $\cA$-symmetric subspace~$(\cH \ot \cH')^{\U(\cA)}$ takes the form
\begin{align}\label{eq:pi}
    \Pi_\cA = \bigoplus_{\lambda\in\Lambda} \frac{\ket{\Gamma}_{\cL_\lambda\cL'_\lambda}\bra{\Gamma}_{\cL_\lambda\cL'_\lambda}}{\dim\cL_\lambda} \ot I_{\cR_\lambda} \ot I_{\cR'_\lambda}.
\end{align}
Thus, \eqref{eq:def in proof} shows that the output states of~$\cP_\lambda$ are supported on this subspace.
Furthermore, the second equation in~\eqref{eq:sym action} shows that the output states commute with~$(I \ot g\tran)$ for all~$g \in G$.

\item
Up to the decompositions of~$\cH$ and of~$\cH'$ (which can be obtained from the one for~$\cH$ and the unitary~$\ket x \mapsto \ket{x'}$, see footnote~\ref{footnote:HHprime}), the channel~\eqref{eq:def in proof} can be implemented by the following efficient steps:
given as input an arbitary quantum state~$\rho$,
\begin{enumerate}[noitemsep]
    \item Measure~$\lambda$ to obtain the (unnormalized) post-measurement state $P_\lambda \rho P_\lambda$ on~$\cL_\lambda \ot \cR_\lambda$.
    \item Discard the register~$\cL_\lambda$, while keeping the register~$\cR_\lambda$.
    \item Prepare the maximally entangled state $\frac1{\sqrt{\dim\cL_\lambda}} \ket{\Gamma}_{\cL_\lambda\cL_\lambda'}$, with $\cL_\lambda'$ an additional register.
    \item Prepare the maximally mixed state on an additional register~$\cR_\lambda'$.
\end{enumerate}

\item[4.]
We need to verify the following formulas:
\begin{align*}
    \cP_\cA[\rho]
=   \Pi_\cA \parens*{ \rho \ot \Omega } \Pi_\cA
=   \sqrt{\cP_\cA[I]} \parens*{ \rho \ot I } \sqrt{\cP_\cA[I]}.
\end{align*}
The first equality is immediate from \cref{eq:rho,eq:pi,eq:def in proof}.
For the second equality, we need that
\begin{align*}
    \parens*{ I \ot \sqrt\Omega} \Pi_\cA = \sqrt{\cP_\cA[I]};
\end{align*}
and indeed both are equal to
\begin{align*}
    \bigoplus_{\lambda\in\Lambda} \frac{\ket{\Gamma}_{\cL_\lambda\cL'_\lambda}\bra{\Gamma}_{\cL_\lambda\cL'_\lambda}}{\sqrt{\dim\cL_\lambda \dim\cR_\lambda}} \ot I_{\cR_\lambda} \ot I_{\cR'_\lambda},
\end{align*}
as follows readily from \cref{eq:pi,eq:def in proof}.
\end{enumerate}
\end{proof}

\subsection{Examples}
We discuss some easy examples of symmetries to illustrate the theorem.
In all examples, we take~$\cH = \cH'$ with the computational basis, and it holds that $\cA = \cA\tran$ and $G = G\tran$ for this choice.

\paragraph{Classical states:}
Consider $\cH = \C^d$, and let~$\cA$ be the algebra consisting of the diagonal matrices.
This algebra is its own commutant, $\cA' = \cA$, and the terms~$L_\lambda \ot R_\lambda$ in the decomposition~\eqref{eq:decomposition} are one-dimensional and spanned by~$\proj x$ for $x \in [d]$.
Note that $\rho \in \cA'$ if and only if $\rho$ is a diagonal density matrix, that is, a classical state.
Moreover, $\E_{\cA\tran}$ is the standard basis measurement channel.
The random purification channel obtained from \cref{thm:main technical} is the classical copying channel:
\[ \cP^\text{Copy}[\rho] = \sum_{x=1}^d \braket{x | \rho | x} \proj{xx}. \]

Next, we consider three examples where $\cH = (\C^d)^{\ot n}$.
The unitary group $\U(d)$ acts by operators~$U^{\ot n}$ for $U \in \U(d)$, while the permutation group~$S_n$ acts by operators~$R_\pi$ for $\pi \in S_n$ that map $R_\pi \ket{x_1,\dots,x_n} = \ket{x_{\pi^{-1}(1)},\dots,x_{\pi^{-1}(n)}}$ for~$x \in [d]^n$.
Clearly, these two actions commute.
Schur-Weyl duality states that they generate each other's commutant.
That is, if we denote by~$\cA_{\U(d)}$ and~$\cA_{S_n}$ the algebras generated by these group actions, then $\cA_{\U(d)}' = \cA_{S_n}$ and vice versa.
We have
\[
    (\C^d)^n \cong \bigoplus_{\lambda} V_{\lambda} \ot W_{\lambda},
\]
where the $V_\lambda$ are irreducible representations of~$\U(d)$ and the $W_\lambda$ are irreducible representations of~$S_n$, labeled by Young diagrams~$\lambda$ with $n$ boxes and no more than~$d$ rows.
The decomposition in \cref{eq:decomposition} for either algebra can be obtained from this at once.

\paragraph{Permutation-invariant states:}
Let $\cA = \cA_{\U(d)}$, so $\cA' = \cA_{S_n}$.
We have that $\rho \in \cA'$ if it is permutation-invariant, that is, such that $R_\pi \rho = \rho R_\pi$ for all~$\pi \in S_n$.
In particular, this is the case if $\rho = \sigma^{\ot n}$ consists of $n$ copies of a state~$\sigma$.
Then the channel in \cref{thm:main technical} is the original random purification channel from~\cite{tang2025conjugate,pelecanos2025mixed,girardi2025random}, which produces a state that is a purification of~$\rho$ with a random unitary tensor power~$U^{\ot n}$ applied, see \cref{cor:og channel}.

\paragraph{Werner states:}
We can also take $\cA = \cA_{S_n}$, so $\cA' = \cA_{\U(d)}$.
Then $\rho \in \cA'$ means that the state is invariant under all product unitaries: $U^{\ot n} \rho = \rho U^{\ot n}$ for all~$U \in \U(d)$.
Such a state is called a (multipartite) Werner state.
The channel in \cref{thm:main technical} produces a state that is a purification of~$\rho$ with a random permutation applied.

\paragraph{Symmetric Werner states:}
Finally we take $\cA$ to be the algebra generated by both~$\cA_{\U(d)}$ and~$\cA_{S_n}$.
Then $\rho \in \cA'$ if it is permutation-invariant as well as unitarily invariant---such states are also called symmetric Werner states.
They can be written as
$\rho = \sum_{\lambda} p_\lambda \tau_{\lambda}$
for some probability distribution $(p_\lambda)$ over Young diagrams~$\lambda$, and $\tau_{\lambda} = P_\lambda / \!\rk P_\lambda$ the maximally mixed state on the isotypical component labeled by~$\lambda$.
The channel in \cref{thm:main technical} measures~$\lambda$ and prepares another copy of~$\tau_\lambda$:
\begin{align*}
    \cP[\rho] = \sum_\lambda P_\lambda \rho P_\lambda \ot \tau_{\lambda}.
\end{align*}

\section{Random purification and tomography for fermionic Gaussian states}\label{sec:fermions}
We now explain how to use \cref{thm:main technical} to obtain a random purification theorem for fermionic Gaussian states---also known as (quasi)free fermionic states---and apply it to construct a tomography protocol that achieves optimal sample complexity.

\subsection{Gaussian states and unitaries}
We start by defining these states and related concepts, see e.g.\ \cite[App.~B]{helsen2022matchgate}.

\paragraph{Fermionic Hilbert space and operators:}
If we have $m$ fermionic modes, then the fermionic Hilbert space is the so-called fermionic Fock space
\begin{align*}
    \cH_m = \bigwedge \C^m = \bigoplus_{k=0}^m \bigwedge^k \C^m. 
\end{align*}
There is a natural algebra on $\cH_m$ generated by the fermionic creation and annihilation operators, denoted $a_j^\dagger$ and~$a_j$ for $j \in [m]$, which satisfy the canonical anticommutation relations
\begin{align*}
    \{a_j, a_k^\dagger\} = \delta_{jk} \id, \qquad \{a_j^\dagger, a_k^\dagger\} = \{a_j, a_k\} = 0.
\end{align*}
We also define the number operators by $N_j = a_j^\dagger a_j$; they satisfy $[N_j, N_k] = 0$.
The Majorana operators~$c_j$ for $j \in [2m]$ are defined by
$c_{2j-1} = a_j^\dagger + a_j$ and $c_{2j} =  i(a_j^\dagger - a_j)$,
which satisfy
\begin{align*}
    \{c_j, c_k\} = 2\delta_{jk} \id.
\end{align*}
The parity operator is defined by $P = i^m c_1 c_1 \dots c_{2m} = (-1)^{\sum_i N_i}$.
It anticommutes with all Majorana operators.
The Hilbert space splits into the~$\pm 1$ eigenspaces of~$P$, called the even and odd parity subspaces~$\cH_m^+$ and~$\cH_m^-$.
The even and odd parity subspaces are spanned by basis states~$\ket x$ with respectively even or odd Hamming weight $\abs{x}$.

The fermionic Fock space be identified with the space of $m$ qubits,
\begin{align*}
    \cH_m \cong (\C^2)^{\ot m}.
\end{align*}
We denote the computational basis by~$\ket x$ for $x\in\bit^m$.
The Jordan-Wigner map allows us to define the fermionic algebra on the $m$-qubit Hilbert space.
For the Majorana operators,
\begin{align*}
    c_{2j-1} \mapsto Z_1 \dots Z_{j-1} X_j, \qquad c_{2j} \mapsto Z_1\dots Z_{j-1} Y_j
\end{align*}
where $P = X,Y,Z$ are the Pauli operators, and $P_j$ denotes the Pauli operator on the $j$th qubit.
Then the number operator $N_j$ acts as $\proj{1}$ on the $j$th qubit, and the parity operator is mapped to
\begin{align*}
    P \mapsto Z_1 Z_2 \cdots Z_m
\end{align*}
by the Jordan-Wigner map.

\paragraph{Gaussian unitaries:}
For $R \in \SO(2m)$, there are unitaries $U_R$ that map
\begin{align}\label{eq:gaussian U}
    U_R c_j U_R^\dagger = \sum_{k=1}^{2m} R_{kj} c_k.
\end{align}
This defines a projective representation\footnote{There is a $\pm 1$ sign ambiguity; the projective representation can be lifted to a genuine representation of $\Spin(2m)$, the double cover of $\SO(2m)$.} of $\SO(2m)$ on $\cH_m$.
The subspaces are $\cH_m^+$ and $\cH_m^-$ are irreducible subrepresentations.
A unitary that acts as in~\eqref{eq:gaussian U} is called a Gaussian unitary.
A quadratic Hamiltonian is a Hamiltonian of the form
\begin{align}\label{eq:quadratic hamiltonian}
    H = i\sum_{j,k=1}^m h_{jk} c_j c_k = 2 \sum_{j<k} h_{jk} c_j c_k,
\end{align}
where $h_{jk} = -h_{kj} \in \RR$ (so $h \in \so(2m)$).
Then $U_R = \exp(i H t)$ is a Gaussian unitary, and all Gaussian unitaries can be written in this way (up to an irrelevant phase).

\paragraph{Gaussian states:}
A pure state on~$\cH_m$ is called Gaussian if it is the unique ground state of a quadratic Hamiltonian~\eqref{eq:quadratic hamiltonian}.
We denote the set of pure Gaussian states by $\cG_m$, and denote the subsets of even and odd parity Gaussian states by $\cG_m^+$ and $\cG_m^-$, respectively.

More generally, mixed Gaussian states are Gibbs states of quadratic Hamiltonians,
\begin{align}\label{eq:fermionic gibbs state}
    \rho = \frac1Z e^{-H}, \quad Z = \tr[e^{-H}],
\end{align}
and limits thereof.
There always exists a Gaussian unitary that transforms the Hamiltonian $H$ to be of the form
\begin{align}\label{eq:std H}
    H = \sum_{j=1}^m \beta_j N_j
\end{align}
(up to an irrelevant additive constant).
This allows one to bring Gaussian states in a standard form.

Furthermore, we can choose some fixed pure Gaussian states~$\ket{\psi_m^{\pm}} \in \cH_m^{\pm}$, which we will take to be highest weight vectors for $\so(2m)$ with respect to an appropriate choice of positive roots, such that all pure Gaussian states can be written as $U_R\ket{\psi_+}$ or $U_R\ket{\psi_-}$ for some Gaussian unitary~$U_R$.
In particular, by taking the Haar measure on~$\SO(2m)$, we obtain unique probability measures on~$\cG_m^+$ and~$\cG_m^-$ that are invariant under Gaussian unitaries, which we will simply denote by $d\phi$.
The highest weights for~$\cH_{\pm}$ are given~by
\begin{align}\label{eq:highest weight}
    \alpha = (\tfrac12, \dots, \tfrac12, \tfrac12), \qquad  \beta = (\tfrac12, \dots, \tfrac12, -\tfrac12)
\end{align}
where for even $m$, $\cG_m^+$ has highest weight $\alpha$ and $\cG_m^-$ has highest weight $\beta$, while for odd $m$ this is reversed.
Details on these representations can be found in \cite{fulton2013representation}.

\subsection{Random purification of Gaussian states}\label{sub:gaussian purification}
We will now show that there exists a random purification channel that maps any mixed Gaussian state to a random Gaussian purification.
This will follow directly from \cref{thm:main technical} but for a small subtlety that we have to address carefully---namely, when we take the tensor product~$\cH_m \ot \cH_m'$ to purify the state, we want to identify this with the Hilbert space~$\cH_{2m}$ of~$2m$ fermionic modes, and we want the purification to be a Gaussian state on~$2m$ modes.
To this end we take $\cH'_m = \cH_m$, so that $\cH_m \ot \cH'_m \cong (\C^2)^{\ot m} \ot (\C^2)^{\ot m} \cong (\C^2)^{\ot 2m} \cong \cH_{2m}$.
Using the Jordan-Wigner map, the Majorana operators~$\tilde c_j$ for $j \in [4m]$ are given by
\begin{align*}
    \tilde c_j &= c_j \ot \id \quad \text{ for } j \in [2m], \\
    \tilde c_{2m + j} &= P \ot c_{j}  \quad \text{ for } j \in [2m].
\end{align*}
We will now choose the following basis on~$\cH'_m$,
\begin{align*}
    \ket{x'} = (-1)^{\gamma(x)} \ket x \qquad (x \in \bit^m),
\end{align*}
where
\begin{align*}
    \gamma(x) = \begin{cases}
        0 & \text{ if } \abs{x} \equiv 0 \text{ or } 1 \pmod 4, \\
        1 & \text{ if } \abs{x} \equiv 2 \text{ or } 3 \pmod 4.
    \end{cases}
\end{align*}
Note that the transpose and conjugate with respect to this basis equals the transpose and conjugate, respectively, with respect to the standard basis.
With respect to our choice of bases, the standard purification of a state~$\rho \in \S(\cH_m)$ is given by
\begin{align*}
    \ket{\psi^\text{std}_\rho}
= \left( \sqrt\rho \ot \id \right) \left( \sum_{x \in \bit^m} \ket{xx'} \right)
= \left( \sqrt\rho \ot \id \right) \left( \sum_{x \in \bit^m} (-1)^{\gamma(x)} \ket{xx} \right).
\end{align*}
We will now show that if $\rho$ is Gaussian, then its standard purification is again a Gaussian state~$\cH_{2m}$.

\begin{lem}\label{lem:gaussian purification}
    Let $\cH_m$ be the Hilbert space of $m$ fermionic modes.
    \begin{enumerate}
        \item If $U$ is a Gaussian unitary, then $\overline{U}$ and $U\tran$ are also Gaussian unitaries. Moreover, $U \ot \id$ and $\id \ot U$ are Gaussian unitaries on $\cH_{2m}$.
        \item If $\rho \in \S(\cH)$ is Gaussian, then $\ket{\psi^\text{std}_\rho}$ is a pure Gaussian state in $\cH_{2m}^+$.
    \end{enumerate}
\end{lem}

\begin{proof}
\begin{enumerate}
\item
    If $U$ is a Gaussian unitary, then we can write $U = \exp(iHt)$ up to some phase, where $H$ is a quadratic Hamiltonian in the Majorana operators.
    Then $\overline{H}$ and $H\tran$ are also quadratic Hamiltonians in the Majorana operators, since for each $c_j$ we have $c_j\tran = \overline{c_j} = \pm c_j$, so $\overline{U}$ and~$U\tran$ are also Gaussian unitaries.
    Additionally, if $H$ is given as in \cref{eq:quadratic hamiltonian}, then, since~$P^2 = \id$,
    \begin{align*}
        \id \ot H = i\sum_{j,k} h_{jk} \id \ot c_j c_k = i\sum_{j,k} h_{jk} (P \ot c_j)(P \ot c_k) =  i\sum_{j,k} h_{jk} \tilde c_{2m+j} \tilde c_{2m+k}
    \end{align*}
    is a quadratic Hamiltonian, and it is also clear that $H \ot \id$ is a quadratic Hamiltonian. This implies that $U \ot \id$ and $\id \ot U$ are Gaussian unitaries.

\item
    Let $\rho \in \S(\cH_m)$ be an arbitrary Gaussian state.
    By a continuity and compactness argument, we may assume without loss of generality that $\rho$ has the form of \cref{eq:fermionic gibbs state}.
    Then there is a Gaussian unitary~$U$ such that (see \cref{eq:std H}),
    \begin{align*}
        \sigma := U \rho U^\dagger = \frac{e^{- H}}{\tr[e^{- H}]}, \qquad H = \sum_{j=1}^m \beta_j N_j.
    \end{align*}
    By the transpose trick, the standard purifications of $\rho$ and $\sigma$ are related by
    \begin{align*}
        \ket{\psi^\text{std}_\rho} = (U^\dagger \ot U\tran) \ket{\psi^\text{std}_\sigma}
    \end{align*}
    The first part of the Lemma implies that~$U^\dagger \ot U\tran$ is a Gaussian unitary, so it suffices to show that $\ket{\psi^\text{std}_\sigma}$ is a Gaussian state.
    To this end, we write
    \begin{align*}
        \sqrt{\sigma} = \frac{\exp(- H/2)}{\sqrt{\tr[\exp(- H])}} = \prod_{i=1}^m \frac{\exp(- \beta_i N_i / 2)}{\sqrt{1 + e^{- \beta_i }}} .
    \end{align*}
    Noting that
    \begin{align*}
        \exp(- \beta_i N_i / 2) = \proj{0}_i + e^{- \beta_i / 2} \proj{1}_i .
    \end{align*}
    we see that
    \begin{align*}
        \ket{\psi^\text{std}_\sigma} = \sum_{x \in \bit^m} \parens*{ \prod_{j=1}^m \cos(\theta)^{1 - x_j}\sin(\theta)^{x_j} } (-1)^{\gamma(x)} \ket{xx}
    \end{align*}
    for angles $\theta_j \in \RR$ chosen such that
    \begin{align*}
        \cos(\theta_j) = \frac{1}{\sqrt{1 + e^{- \beta_j}}}, \qquad \sin(\theta_j) = \frac{e^{- \beta_j / 2}}{\sqrt{1 + e^{- \beta_j}}} \, .
    \end{align*}
    The result now follows from \cref{lem:gaussian diagonal purification}.
\end{enumerate}
\end{proof}

\begin{lem}\label{lem:gaussian diagonal purification}
    Let $\theta_j \in \RR$ for $j \in [m]$.
    Then
    \begin{align*}
        \ket{\psi} = \sum_{x \in \bit^m} \parens*{ \prod_{j=1}^m \cos(\theta_j)^{1 - x_j}\sin(\theta_j)^{x_j} } (-1)^{\gamma(x)} \ket{xx}
    \end{align*}
    is a pure Gaussian state in $\cH_{2m}^+$.
\end{lem}
\begin{proof}
    Let
    \begin{align*}
        \ket{\psi_k} = \sum_{\substack{y \in \bit^k,\\ x = (y,0^{m-k})}} \parens*{ \prod_{j=1}^k \cos(\theta_j)^{1 - x_j}\sin(\theta_j)^{x_j} } (-1)^{\gamma(x)} \ket{xx}.
    \end{align*}
    We will prove by induction that these states are Gaussian for all~$k\in\{0,\dots,m\}$.
    It is clear that $\ket{\psi_0} = \ket{0}^{\ot 2m}$ is Gaussian.
    Now suppose that $\ket{\psi_{k-1}}$ is Gaussian.
    Let $H_k = -i \tilde c_{2k} \tilde c_{2m+2k}$.
    Then
    \begin{align*}
        H_k = -i Y_k Z_{k} Z_{k+1}\dots Z_{m + k - 1} Y_{m+k} = X_k Z_{k+1}\dots Z_{m + k - 1} Y_{m+k}.
    \end{align*}
    We claim that the Gaussian unitary $U_k = \exp(-i H_k \theta_k)$ maps $\ket{\psi_{k-1}}$ to $\ket{\psi_k}$, and hence the latter state is Gaussian as well.
    We write
    \begin{align*}
        U_k \ket{\psi_{k-1}} = \sum_{\substack{y \in \bit^{k-1},\\ x = (y,0^{m-k +1})}} \parens*{ \prod_{j=1}^{k-1} \cos(\theta_j)^{1 - x_j}\sin(\theta_j)^{x_j} } (-1)^{\gamma(x)}\left(\cos(\theta_k) \id - (-1)^{\abs{x}} i \sin(\theta_k) X_k Y_{m+k} \right)\ket{xx}.
    \end{align*}
    We rewrite this as a sum over $\tilde y \in \bit^k$, and $\tilde x = (\tilde y,0^{m-k})$.
    The first term in the right-hand side sum corresponds to~$\tilde y = (y,0)$, and the second one to~$\tilde y = (y,1)$.
    In the former case ew have $\ket{xx} = \ket{\tilde x \tilde x}$, while in the latter case we have $X_k Y_{m+k} \ket{xx} = i \ket{\tilde x \tilde x}$, so we can rewrite this as
    \begin{align*}
        U_k \ket{\psi_{k-1}} = \sum_{\substack{\tilde y \in \bit^{k},\\ \tilde x = (\tilde y,0^{m-k})}} \parens*{ \prod_{j=1}^{k-1} \cos(\theta)^{1 - x_j}\sin(\theta_j)^{x_j} } \alpha(\tilde x) \ket{\tilde x \tilde x}.
    \end{align*}
    where $\alpha(\tilde x) = (-1)^{\gamma(x) + \tilde x_k \abs{x}}$.
    It is easy to verify that $\alpha(\tilde x) = (-1)^{\gamma(\tilde x)}$ by a case distinction on~$\tilde x_k \in \bit$.
    We conclude that $U_k \ket{\psi_{k-1}} = \ket{\psi_k}$.
    This concludes the proof, since $\ket{\psi_m} = \ket{\psi}$.
\end{proof}

The final ingredient we need is that the Gaussian states have the right symmetry property.
The following lemma avoid appealing to Howe duality~\cite[Thm.~34]{girardi2025gaussian}.

\begin{lem}\label{lem:gaussian state in algebra}
    If $\sigma$ is a $m$-mode fermionic Gaussian state, then $\sigma^{\ot n} \in \C\{U_R^{\ot n}  :  R \in \SO(2m) \}$.
\end{lem}

\begin{proof}
    It suffices to consider the dense set of fermionic Gaussian states of the form~\eqref{eq:fermionic gibbs state}, with Hamiltonian~$H$.
    Let $\cB = \C\{U_R^{\ot n}  : R \in \SO(2m) \}$.
    Then we have $V(t) = \exp(iHt)^{\ot n} \in \cB$, and hence also
    \begin{align*}
         -i\partial_t \big|_{t=0} V(t) = \sum_{j=1}^n H_j \in \cB
    \end{align*}
    where $H_j$ is the operator that acts as $H$ on the $j$-th copy of $\cH_m$ and as $\id$ on the other copies.
    This implies that $\exp(-H) \in \cB$ as well.
\end{proof}

We thus obtain the Gaussian random purification result as stated in the introduction.

\begin{proof}[Proof of \cref{cor:fermions}]
This is now an immediate consequence of \cref{thm:main technical}, together with \cref{lem:gaussian state in algebra}, and the fact, proven in \cref{lem:gaussian purification}, that the standard purification is a Gaussian state.
\end{proof}


\subsection{Tomography and propery testing of pure Gaussian states}\label{sub:gaussian pure tomo}
We first propose a tomography scheme for pure Gaussian states, analogous to Hayashi's POVM for ordinary pure state tomography~\cite{hayashi1998asymptotic}.
Recall that~$\cG_m^{\pm}$ denotes the set of even resp.\ odd parity pure Gaussian states.
For a number of copies~$n$, let $\cV_{n,m}^\pm$ denote the span of~$\ket{\psi}^{\ot n}$ for~$\ket{\psi} \in \cG_m^{\pm}$.
Recall that $\cV_{1,m}^\pm = \cH_m^\pm$ is an irreducible representation of~$\so(2m)$ with highest weight~$\alpha$ or~$\beta$ (depending on whether~$m$ is even or odd, see \cref{eq:highest weight}) and highest weight vector $\ket{\psi_m^\pm} \in \cG_m^{\pm}$.
It follows that, for every~$n$, $\cV_{n,m}^\pm$ is an irreducible $\so(2m)$-representation with highest weight~$n\alpha$ or~$n\beta$, and highest weight vector~$\ket{\psi_m^{\pm}}^{\ot n}$.

\begin{lem}\label{lem:dimension multicopy fermions}
    We have $\dim \cV_{n,m}^+ = \dim \cV_{n,m}^- = d_{n,m}$, where
    \begin{align}\label{eq:dimension multicopy fermions}
        d_{n,m} = \prod_{1 \leq j < k \leq m} \frac{2m + n - (j+k)}{2m - (j+k)}
    \end{align}
\end{lem}
\begin{proof}
This is a consequence of the Weyl character formula.
For $\so(2m)$, it gives the following formula for the dimension of an irreducible representation with highest weight~$\lambda$~\cite[Eq.~(24.41)]{fulton2013representation}:
\begin{align*}
    d_{n,m} = \prod_{1 \leq j < k \leq m} \frac{L_j^2 - L_k^2}{K_j^2 - K_k^2} \qquad K_j = m - j, \ L_j = \lambda_j + m - j.
\end{align*}
Note that $L_m^2 = \lambda_m^2$ does not depend on the sign of $\lambda_m$, and in particular $\dim \cV_+ = \dim \cV_-$.
By substiting~$\lambda = n\alpha$ or~$n\beta$, we obtain \cref{eq:dimension multicopy fermions}.
\end{proof}

Let $\dd \phi$ denote the unique probability measure on~$\cG_m^\pm$ that is invariant under Gaussian unitaries.
Then the operator
\begin{align*}
    \Pi_{n,m}^\pm = d_{n,m} \, \int_{\cG_m^{\pm}} \dd \phi \, \proj{\phi}^{\ot n}
\end{align*}
is easily seen to commute with Gaussian unitaries, hence is an intertwiner for the action of~$\so(2m)$.
By Schur's lemma, $\Pi_{n,m}^\pm$ must be proportional to the projection on $\cV_{n,m}^\pm$.
By taking the trace, we see that~$\Pi_{n,\pm}$ is exactly equal the projection operator on $\cV_{n,m}^\pm$.
This implies that $\mu_{n,\pm}$ defined by
\begin{align}\label{eq:haar measure projection}
    \dd \mu_{n,m}^\pm(\phi) = d_{n,m} \,\dd \phi \, \proj{\phi}^{\ot n},
\end{align}
where $d_{n,m}$ is given in \cref{eq:dimension multicopy fermions}, is a POVM on~$\cV_{n,m}^\pm$ with outcomes in~$\cG_m^\pm$.
Of course, we can combine $\mu_{n,m}^\pm$ into a single POVM~$\mu_{n,m}$ on~$\cV_{n,m}^+ \oplus \cV_{n,m}^-$ with outcomes in~$\cG_m$.

\begin{lem}\label{lem:overlap bound free fermions}
Let $\psi \in \cG_m$ be a pure fermionic Gaussian state.
Let $\hat\psi$ denote the outcome of the measurement of the POVM~$\mu_{n,m}$ given $n$ copies of~$\psi$, i.e., $\psi^{\ot n}$.
Then the expected squared overlap can be bounded by
\begin{align*}
    \Ex \abs{\braket{\psi | \hat{\psi}}}^2 \geq \left(1 - \frac{1}{n} \right)^{m(m-1)/2}.
\end{align*}
\end{lem}

\begin{proof}
    We consider the case where $\ket\psi \in \cG_m^+$ has even parity (the odd case is identical):
    \begin{align*}
        \Ex \abs{\braket{\psi | \hat{\psi}}}^2
        &= \int_{\phi \in \cG_m^+} \tr[ \dd \mu_{n,m}(\phi) \proj{\psi}^{\ot n} ] \ \abs{\braket{\psi|\phi}}^2 \\
        &= d_{n,m} \int_{\phi \in \cG_m^+} \dd\phi \, \tr[ \proj\phi^{\ot n} \proj{\psi}^{\ot n} ] \ \abs{\braket{\psi|\phi}}^2 \\
        &= d_{n,m} \int_{\phi \in \cG_m^+} \dd\phi \, \abs{\braket{\psi|\phi}}^{2n+2} \\
        &= \frac {d_{n,m}} {d_{n+1,m}} \int_{\phi \in \cG_m^+} \tr[ \dd \mu_{n+1,m}(\phi) \proj{\psi}^{\ot (n+1)} ] \\
        &= \frac {d_{n,m}} {d_{n+1,m}} \bra\psi^{\ot (n+1)} \Pi_{n,m}^\pm \ket\psi^{\ot (n+1)} ] \\
        &= \frac {d_{n,m}} {d_{n+1,m}}.
    \end{align*}
    \Cref{lem:dimension multicopy fermions} allows us to compute and bound the ratio of dimensions:
\begin{align*}
    \frac{d_{n,m}}{d_{n+1,m}} &= \prod_{1 \leq i < j \leq m} \frac{2m + n - (i+j)}{2m + n + 1 - (i+j)} = \prod_{1 \leq i < j \leq m} \left(1 - \frac{1}{2m + n + 1 - (i+j)}\right) \\
    &\geq \left(1 - \frac{1}{n} \right)^{m(m-1)/2}.
\end{align*}
\end{proof}

\begin{cor}[Tomography of pure fermionic Gaussian states]\label{cor:pure fermion tomography}
    Let $\eps,\delta>0$.
    Given $n = \bigO(m^2/(\eps\delta))$ copies of a pure fermionic Gaussian state~$\psi \in \cG_m$, the measurement~$\mu_{n,m}$ outputs an estimate~$\hat\psi$ such that~$\abs{\braket{\psi|\hat\psi}}^2 \geq 1 - \eps$, with probability at least~$1-\delta$.
\end{cor}
\begin{proof}
    By Markov's inequality and \cref{lem:overlap bound free fermions},
    \begin{align*}
        \Pr\mleft( \abs{\braket{\psi|\hat\psi}}^2 \geq 1 - \eps \mright)
    \leq \frac{1 - \Ex \abs{\braket{\psi|\hat\psi}}^2}\eps
    \leq O\mleft(\frac{m^2}{n\eps}\mright).
    \end{align*}
    Thus, $n = \bigO(m^2/(\eps\delta))$ copies suffice to make the right-hand side~$\leq\delta$.
\end{proof}

As discussed in the introduction, this result improves the sample complexity of pure fermionic Gaussian tomography to~$\bigO(m^2)$, which is optimal.
As a direct corollary, we can also find an improved algorithm for pure fermionic \emph{Gaussianity testing}---the property testing problem of deciding if an unknown fermionic pure state is Gaussian or far from it, given a number of copies of the state.



\begin{proof}[Proof of \cref{cor:testing fermions}]
    Consider the following protocol:
    \begin{enumerate}[noitemsep]
        \item First use $n_1 = \bigO(m^2 / \eps)$ copies for tomography using the measurement $\mu_{n_1,m}$ in \cref{cor:pure fermion tomography} to obtain an estimate $\hat\psi$ such that, if $\ket{\psi}$ is Gaussian, then $\abs{\braket{\psi | \hat{\psi}}}^2 \geq 1 - \eps/2$ holds with constant probability of success.
        \item Then prepare $n_2 = \Theta(1/\eps)$ copies of $\hat\psi$ and use $n_2$ additional copies of $\ket{\psi}$ to perform $n_2$ swap tests. Accept if all swap tests accept.
    \end{enumerate}
    We now analyze the correctness of this algorithm.
    If $\psi$ is Gaussian, then with constant probability, $\hat{\psi}$ has overlap~$\abs{\braket{\psi | \hat{\psi}}}^2 \geq 1 - \eps/2$.
    If, on the other hand, the state $\psi$ is $\eps$-far from Gaussian, then (since $\hat\psi$ is Gaussian) we have $\abs{\braket{\psi | \hat{\psi}}}^2 \leq 1 - \eps$.
    When choosing $n_2 = \Theta(1/\eps)$ appropriately, the swap test distinguishes these two cases with constant probability of error.
\end{proof}

In \cref{cor:testing fermions} we consider the standard ``one-sided'' property testing scenario where the given state is exactly Gaussian in the accepting case.
The previous best result is given in~\cite{bittel2025optimalfermion}.
To compare, we take $\eps_A = 0$ and $\eps = \eps_B^2$ in the notation of their work, which gives a sample complexity of $n = \bigO(m^5/\eps^2)$ in \cite[Thm.~13]{bittel2025optimalfermion}.

While our sample complexity is significantly improved, it is unlikely to be optimal.
Indeed, it does not only test whether a state is Gaussian, but also learns a description of the state if the state is Gaussian.
Since $\cV_{n,m}^\pm$ is the span of $\ket{\psi}^{\ot n}$ for Gaussian states~$\psi \in \cG_m^\pm$, the projection $\Pi_{n,m} = \Pi_{n,m}^+ + \Pi_{n,m}^-$ is the minimal POVM element on $\cH_m^{\ot n}$ that accepts $\psi^{\ot n}$ for any pure fermionic Gaussian state.
Accordingly, the corresponding projective measurement is the optimal test on $n$ copies that is perfectly complete.
We conjecture that with this measurement, pure Gaussianity testing for fermions can be done using a number of copies independent of~$m$, similar to the case of stabilizer states~\cite{gross2021schur}.
See also \cite{girardi2025gaussian} for recent progress in the bosonic setting.

\subsection{Tomography of mixed Gaussian states}
Finally, we can combine the random Gaussian purification channel (\cref{sub:gaussian purification}) with our result on tomography result for Gaussian pure states (\cref{sub:gaussian pure tomo}) to obtain a tomography algorithm for Gaussian mixed states with quadratic sample complexty.

\begin{proof}[Proof of \cref{thm:tomo}]
The tomography algorithm proceeds as follows:
\begin{enumerate}[noitemsep]
\item First apply the random purification channel $\cP^{\text{Fermi}}_{n,m}$ of \cref{cor:fermions} to $\sigma^{\ot n}$
\item Then perform the pure-state tomography protocol of \cref{cor:pure fermion tomography} on the output, to obtain an estimate $\hat{\psi}$ of a $2m$-mode pure fermionic Gaussian state.
\item Return $\hat\sigma$, the partial trace of $\hat\psi$ over the purifying system.
\end{enumerate}
The analysis is completely straightforward.
By \cref{cor:fermions}, $\cP^{\text{Fermi}}_{n,m}$ has the interpretation of preparing $n$ copies of a random Gaussian purification of $\sigma$.
By \cref{cor:pure fermion tomography}, it follows that the state $\psi$ satisfies $\abs{\braket{\psi | \hat \psi}}^2 \geq 1 - \eps$ for \emph{some} Gaussian purification~$\ket{\psi}$ of~$\sigma$, with constant probability.
By taking a partial trace we get obtain an estimate~$\hat\sigma$ such that $F(\sigma,\hat \sigma)^2 \geq \abs{\braket{\psi | \hat \psi}}^2 \geq 1 - \eps$.
\end{proof}

For arbitrary (full-rank) states on $\C^d$, the purification has $O(d^2)$ degrees of freedom.
In that case, tomography for pure states has a sample complexity scaling with $d$, while for mixed states it scales with $d^2$.
In the fermionic case, the Hilbert space purifying an $m$-mode fermionic system has~$2m$ modes.
This again squares the Hilbert space dimension, but since an $m$-mode fermionic state state is specified by~$O(m^2)$ parameters, it is natural to find that fermionic Gaussian tomography has a sample complexity scaling with $m^2$ for both pure and mixed states.

\subsection{Sample complexity lower bound}
We now show that $\bigO(m^2/\eps)$ is in fact the \emph{optimal} sample complexity.
Since we found the same scaling of the sample complexity for pure and mixed states, it suffices to prove a lower bound for estimating pure states, which simplifies matters significantly.
We will follow the argument in~\cite{scharnhorst2025optimal} (see also \cite{harrow2013church}) for lower bounding the sample complexity of (pure) state tomography for arbitrary states, which only requires minor modifications.

Assume that we have an algorithm that takes $n$ copies of an even-parity pure Gaussian fermionic state $\psi \in \cG_m^+$, and returns an estimate $\ket{\hat{\psi}} \in \cH_m^+$.
We can assume without loss of generality that the estimate is itself a pure even parity Gaussian fermionic state.
This corresponds to some POVM~$\mu$, and we get a probability distribution $P_{\psi}$ for the measurement outcomes~given~by
\begin{align*}
    \dd P_{\psi}(\phi) = \tr[\dd \mu(\phi) \proj{\psi}^{\ot n}]
\end{align*}
when applied to $n$ copies of the state~$\psi$.
Given such a measurement, we now bound its average performance when the target state is chosen uniformly at random.

\begin{lem}\label{lem:moments fidelity}
    Suppose that $\psi$ is distributed according to the Haar measure on $\cG_m^+$ or on $\cG_m^-$?, and $\hat{\psi}$ denotes the estimate from the measurement~$\mu$ on $n$ copies of $\psi$.
    Then,
    \begin{align*}
        \Ex_{\psi} \, \Ex_{\hat\psi \sim P_{\psi}} \, \abs{\braket{\psi | \hat\psi}}^{2k} \leq \frac{d_{n,m}}{d_{n+k,m}},
    \end{align*}
    with $d_{n,m}$ given in \cref{eq:dimension multicopy fermions}.
\end{lem}

\begin{proof}
    We only consider the even-parity case since the analysis of the odd-parity case proceeds in the same way.
    We compute
    \begin{align*}
        \Ex_{\psi} \, \Ex_{\hat\psi \sim P_{\psi}} \, \abs{\braket{\psi | \hat\psi}}^{2k}
        &= \int_{\cG_m^+} \int_{\cG_m^+} \tr\bigl[ \dd \mu(\hat\psi) \proj{\psi}^{\ot n} \bigr] \, \abs{\braket{\psi | \hat\psi}}^{2k} \dd \psi \\
        &= \int_{\cG_m^+} \int_{\cG_m^+} \tr\bigl[ (\dd \mu(\hat\psi) \ot \proj{\hat{\psi}}^{\ot k}) \proj{\psi}^{\ot(n+k)} \bigr] \dd \psi \\
        &= \frac{1}{d_{n+k,m}} \int_{\cG_m^+} \tr\bigl[ (\dd \mu(\hat\psi) \ot \proj{\hat{\psi}}^{\ot k}) \Pi_{n+k,m}^+ \bigr],
    \end{align*}
    where in the last equality we use \cref{eq:haar measure projection}.
    The result now follows from
    \begin{align*}
        \int_{\cG_m^+} \tr[(\dd \mu(\hat\psi) \ot \proj{\hat{\psi}}^{\ot k}) \Pi_{n+k,m}^+] &\leq \tr[\Pi_{n,m}^+ \ot \proj{\hat{\psi}}^{\ot k})]
        = \tr[\Pi_{n,m}^+] = d_{n,m},
    \end{align*}
    where we use that $\Pi_{n+k,m}^+ \leq \id$ and we may assume without loss of generality that $\mu$ is supported on $\cV_{n,m}^+$, the space spanned by $n$~identical copies of even parity Gaussian fermionic states.
\end{proof}

We now use this to bound the required number of copies for tomography of Gaussian fermionic states.
We now state and prove a more precise version of \cref{thm:tomo lower bound}.


\begin{proof}[Proof of \cref{thm:tomo lower bound}]
It suffices to restrict to even parity states.
    Let $\mu$ be an $n$-copy measurement, for which we assume that for any $\psi \in \cG_m^+$, the resulting estimate $\hat{\psi}$ satisfies $\abs{\braket{\psi|\hat\psi}}^2 \geq 1 - \eps$ with probability at least 0.99.
        Together with the fact that $\abs{\braket{\psi | \hat\psi}}^{2k} \geq 0$, we have that for every $\psi \in \cG_m^+$,
    \begin{align*}
        \Ex_{\hat\psi \sim P_{\psi}} \, \abs{\braket{\psi | \hat\psi}}^{2k} \geq  0.99(1-\eps)^{2k} .
    \end{align*}
    On the other hand, by \cref{lem:moments fidelity} and \cref{lem:dimension multicopy fermions} we have
    \begin{align*}
        \Ex_{\psi} \, \Ex_{\hat\psi \sim P_{\psi}} \, \abs{\braket{\psi | \hat\psi}}^{2k} &\leq \frac{d_{n,m}}{d_{n+k,m}}
        = \prod_{1 \leq i < j \leq m} \frac{2m + n - (i+j)}{2m + n + k - (i+j)}\\
        &= \prod_{1 \leq i < j \leq m} \left(1 - \frac{k}{2m + n + k - (i+j)}\right)
        \leq \left(1 - \frac{k}{2m + n + k}\right)^{m(m-1)/2}.
    \end{align*}
    We conclude that for every $k$
    \begin{align*}
        0.99(1-\eps)^{2k} \leq \Ex_{\psi} \, \Ex_{\hat\psi \sim P_{\psi}} \, \abs{\braket{\psi | \hat\psi}}^{2k} \leq \left(1 - \frac{k}{2m + n + k}\right)^{m(m-1)/2}.
    \end{align*}
    Using $1 + xy \leq (1+x)^y \leq e^{xy}$, which holds for $x \geq -1$, we deduce that
    \begin{align*}
        0.99(1-2k\eps) \leq \exp\mleft( -\frac{km(m-1)}{2(2m + n + k)}\mright).
    \end{align*}
    We now choose $k = \lfloor 1/(4\eps) \rfloor$, so $0.99(1-2k\eps) \geq e^{-1/2}$ and taking logarithms we get
    \begin{align*}
        \frac{km(m-1)}{2(2m + n + k)} \leq \frac12,
    \end{align*}
    which we can rewrite as
    \begin{equation*}
        n \geq k(m^2 - m - 1) - 2m = \Omega(m^2/\eps).
        \qedhere
    \end{equation*}
\end{proof}

\section{Random purification for bosonic Gaussian states}\label{sec:bosons}
In this section we explain how to use \cref{thm:main technical} to obtain a random purification theorem for gauge-invariant bosonic Gaussian states.

\subsection{Gaussian states and unitaries}\label{sub:boson prelims}
We start by defining these states and related concepts, see e.g.\ \cite{weedbrook2012gaussian}.

\paragraph{Bosonic Hilbert space and operators:}
If we have $m$ bosonic modes, the associated Hilbert space is given by
\begin{align*}
    \cH_m = \mathrm \L^2(\R^m).
\end{align*}
There is a natural algebra on $\cH_m$ generated by the bosonic creation and annihilation operators, denoted $a_j^\dagger$ and~$a_j$ for $j \in [m]$, which satisfy the canonical commutation relations
\begin{align*}
    [a_j, a_k^\dagger] = \delta_{jk} \id, \qquad [a_j^\dagger, a_k^\dagger] = [a_j, a_k] = 0.
\end{align*}
The associated quadrature operators are the position and momentum operators, given by~$q_j = (a_j^\dagger + a_j)/\sqrt2$ and~$p_j = i(a_j^\dagger - a_j)/\sqrt2$ for~$j\in[m]$.
We also define the (total) number operator by~$N = \sum_j a_j^\dagger a_j$.
The associated unitaries~$e^{i\theta N}$ are called gauge transformations; it holds that $e^{i\theta N} a_j e^{-i\theta N} = e^{-i\theta} a_j$ for~$j\in[m]$.
The parity operator is defined as~$P = (-1)^N$.
It anticommutes with all creation and annihilation operators.
The Hilbert space~$\cH_m$ splits into the~$\pm 1$ eigenspaces of~$P$, called the even and odd parity subspaces~$\cH_m^+$ and~$\cH_m^-$.

\paragraph{Gaussian unitaries:}
For $S \in \Sp(2m)$, there are unitaries $U_S$ that map
\begin{align*}
    U_S r_j U_S^\dagger = \sum_{k=1}^{2m} S_{kj} r_k,
\end{align*}
where $(r_1,\dots,r_{2m}) = (q_1,p_1,\dots,q_m,p_m)$ are the quadrature operators.
This defines a projective representation\footnote{The projective representation can be lifted to a genuine representation of the metaplectic group~$\Mp(2m)$, the double cover of $\Sp(2m)$.} of $\Sp(2m)$ on~$\cH_m$.
The subspaces are $\cH_m^+$ and $\cH_m^-$ are irreducible subrepresentations.
A unitary is called a Gaussian unitary if it is a product of a unitary as in~\eqref{eq:gaussian U} with an element of the Heisenberg-Weyl group (that is, a unitary generated by a displacement operator).
Thus, the Gaussian unitaries form a projective representation of the affine symplectic group~$\ASp(2m)$.

A quadratic Hamiltonian is a Hamiltonian~$H$ that is at most quadratic in the creation and annihilation operators (equivalently, in the quadrature operators).
Then $U_S = \exp(i H t)$ is a Gaussian unitary, and all Gaussian unitaries can be written in this way (up to an irrelevant phase).
We call $H$ gauge-invariant (or number-preserving) if it commutes with all gauge transformations (equivalently, with the number operator).
Such Hamiltonians are of the form
\begin{align}\label{eq:quadratic gauge invariant hamiltonian}
    H = \sum_{j,k=1}^m h_{jk} a_j^\dagger a_k
\end{align}
where $h = h^\dagger$ (so $ih$ is in the Lie algebra~$\mathfrak{u}(m))$, up to an irrelevant additive constant.
The corresponding Gaussian unitaries are called passive (or gauge-invariant or number-preserving) Gaussian unitaries.
They form a projective representation of the unitary group~$\U(m)$, which is the maximally compact subgroup of~$\Sp(2m)$.
For any number-preserving Hamiltonian~\eqref{eq:quadratic gauge invariant hamiltonian}, there exists a passive Gaussian unitary transforming it into
\begin{align}\label{eq:diag gauge invariant hamiltonian}
    H = \sum_{j=1}^m \beta_j a_j^\dagger a_j + c
\end{align}
for $\beta_j \in \R$ and $c \in \R$.


\paragraph{Gaussian states:}
A Gaussian state is a Gibbs state of a quadratic Hamiltonian that is bounded from below,
\begin{align}\label{eq:bosonic gibbs state}
    \rho = \frac1Z e^{-H}, \quad Z = \tr[e^{-H}],
\end{align}
or a limit thereof.
There always exists a Gaussian unitary that transforms the Hamiltonian~$H$ into a sum of harmonic oscillators
\begin{align}\label{eq:std H boson}
    H = \sum_{j=1}^m \beta_j \left( a_j^\dagger a_j + \frac12 \right)
\end{align}
where $\beta_j > 0$ (up to an irrelevant additive constant).

A gauge-invariant (or number-preserving) Gaussian state is a Gaussian state that commutes with all gauge transformations (equivalently, with the number operator).
In this case, $H$ must be gauge-invariant and it can be brought into the form \cref{eq:std H boson} by a passive Gaussian unitary.
By comparing
\cref{eq:diag gauge invariant hamiltonian,eq:std H boson}, we see that any gauge-invariant Gaussian state~$\rho$ is in the von Neumann algebra generated by the passive Gaussian unitaries:
$\rho \in \{U_O : O \in \U(m)\}''$.

If we use the basis of occupation number states $\ket{x} = \ket{x_1,\dots,x_m}$ for $\L^2(\R^m)$ to define the standard purification, then we note that
\begin{align*}
    \ket{\psi^\text{std}_\rho} = (\sqrt{\rho} \ot \id) \sum_{x} \ket{xx}
\end{align*}
is a Gaussian state.
To see this, note that by the transpose trick (and the fact that $U_S^{\tran}$ is Gaussian if $U_S$ is Gaussian), we may assume without loss of generality that $\rho$ is the Gibbs state of a Hamiltonian as in \cref{eq:std H boson}, so $\rho$ is a tensor product of $m$ bosonic modes, and the standard purification is a tensor product of states proportional to
\begin{align*}
    \sum_{x_j=0}^{\infty} e^{-\beta_j x_j/2} \ket{x_j x_j},
\end{align*}
where $\ket{x}$ are the occupation number states, and this is a Gaussian state.

\subsection{Random purification of gauge-invariant Gaussian states}
We now prove \cref{cor:bosons}, which asserts the existence of a random purification channel for gauge-invariant bosonic Gaussian states.
One approach would be to formulate an appropriate infinite-dimensional variant of \cref{thm:main simplified,thm:main technical}.
However, we find that this is not necessary; \cref{cor:bosons} can be obtained as a direct consequence of the finite-dimensional result.

\begin{proof}[Proof of \cref{cor:bosons}]
    We decompose the $m$-mode Hilbert space into the eigenspaces of the number operator:
    \begin{align}\label{eq:particle number decomposition}
        \L^2(\R^m) \cong \bigoplus_{k=0}^\infty \cH_k, \qquad \cH_k = \Sym^k(\C^m).
    \end{align}
    The subspaces $\cH_k$ are invariant and indeed irreducible under passive Gaussian unitaries.
    Indeed, a passive Gaussian unitary~$U_O$ for $O \in \U(m)$ acts by~$O^{\ot k}$ on~$\cH_k$ (up to a phase independent of~$k$).
    As a basis (for defining the standard purification) we choose the occupation (or weight) basis, which is compatible with this decomposition.
    We now write
    \begin{align}\label{eq:particle number multicopy}
        \L^2(\R^m)^{\ot n} \cong \bigoplus_{k_1, \dots, k_n = 0}^\infty \cH_{k_1} \ot \cdots \ot \cH_{k_n} = \bigoplus_{k=0}^n \cH_k^{(n)}, \qquad
        \cH_k^{(n)} = \bigoplus_{k_1 + \dots + k_n = k} \cH_{k_1} \ot \cdots \ot \cH_{k_n},
    \end{align}
    and apply \cref{thm:main simplified} to each $\cH_k^{(n)}$ separately, with $G = G_k^{(n)} = \{ U_O^{\ot n}|_{\cH_k^{(n)}} : O \in \U(m) \}$, to obtain random purification channels~$\cP_k^{(n)} \colon \Lin(\cH_k^{(n)}) \to \Lin(\cH_k^{(n)} \ot \cH_k^{(n)})$ with the property that
    \begin{align}\label{eq:cPkn}
        \cP_k^{(n)}[\rho_k^{(n)}] = \int_{O \in \U(m)} \left( \id \ot U_O^{\ot n} \right) \psi^\text{std}_{\rho_k^{(n)}} \left( \id \ot (U_O^\dagger)^{\ot n} \right) \dd O
    \end{align}
    for all states $\rho_k^{(n)} \in \C G^{(n)}_k$.

    We can then define a gauge-invariant bosonic Gaussian purification channel as follows:
    \begin{align*}
        \cP_{n,m}^\text{Boson} = \bigoplus_{k=0}^{\infty} \cP_k^{(n)}, \quad\text{that is,}\quad
        \cP_{n,m}^\text{Boson}[\rho] = \bigoplus_{k=0}^{\infty} \cP_k^{(n)}\bigl[ P_k^{(n)} \rho P_k^{(n)} \bigr]
    \end{align*}
    for every state~$\rho$ on~$\L^2(\R^m)^{\ot n}$, where $P_k^{(n)}$ denotes the orthogonal projection onto $\cH_k^{(n)} \subseteq \L^2(\R^m)^{\ot n}$.
    The output can be interpreted as density operator on~$\L^2(\R^m)^{\ot n} \ot \L^2(\R^m)^{\ot n} \cong \L^2(\R^{2m})^{\ot n}$.

    We will now prove that this channel has the desired property when acting on~$\sigma^{\ot n}$ for any gauge-invariant Gaussian state~$\sigma$.
    Since $\sigma$ is gauge-invariant, $\sigma^{\ot n} \in \C\{U_O^{\ot n} : O \in \U(m)\}$.
    Thus,
    \begin{align*}
        \sigma^{\ot n} \cong \bigoplus_{k=0}^\infty \rho_k^{(n)},
    \end{align*}
    is block-diagonal with respect to \cref{eq:particle number multicopy} and each $\rho_k^{(n)} \in \C G_k^{(n)}$.
    Note that, even though the Hilbert spaces~$\L^2(\R^m)$ and~$\L^2(\R^m)^{\ot n}$ are infinite-dimensional, the standard purifications are well-defined and (recall that we work with respect to the occupation number bases) can be decomposed~as
    \begin{align*}
        \ket{\psi^\text{std}_\sigma}^{\ot n} = \ket{\psi^\text{std}_{\sigma^{\ot n}}} \cong \bigoplus_{k=0}^\infty \ket{\psi^\text{std}_{\rho_k^{(n)}}}
    \in \bigoplus_{k=0}^\infty \cH_k^{(n)} \ot \cH_k^{(n)} \subseteq L^2(\R^{2m})^{\ot n}.
    \end{align*}
    We now compute the twirl over the passive Gaussian unitaries:
    \begin{align*}
        \int_{O \in \U(m)} \parens*{ I \ot U_O^{\ot n} } \psi^{\text{std},\ot n}_\sigma \parens*{ I \ot U_O^{\dagger, \ot n} } \ \dd O
    &\cong  \int_{O \in \U(m)} \parens*{ I \ot U_O^{\ot n} } \bigoplus_{k,\ell=0}^\infty \ketbra{\psi^\text{std}_{\rho_k^{(n)}}}{\psi^\text{std}_{\rho_\ell^{(n)}}} \parens*{ I \ot U_O^{\dagger, \ot n} } \ \dd O \\
    &=  \bigoplus_{k=0}^\infty \int_{O \in \U(m)} \parens*{ I \ot U_O^{\ot n} } \psi^\text{std}_{\rho_k^{(n)}} \parens*{ I \ot U_O^{\dagger, \ot n} } \dd O \\
    &=  \bigoplus_{k=0}^{\infty} \cP_k^{(n)}\bigl[ \rho_k^{(n)} \bigr] = \cP_{n,m}^\text{Boson}[\sigma^{\ot n}],
    \end{align*}
    where the second line follows because $\exp(i N \theta)$ is a passive Gaussian unitary (corresponding to~$e^{i\theta} I \subseteq \U(m)$) and twirling over this one-parameter subgroup is equivalent to measuring the total particle number~$k$, and for the third line we use \cref{eq:cPkn} for $\rho_k^{(n)} \in \C G_k^{(n)}$ and then the definition of the bosonic Gaussian purification channel.
    Since the standard purification is a Gaussian state, as discussed in \cref{sub:boson prelims}, the results follows.
\end{proof}

\section*{Acknowledgments}
We would like to thank Ben Lovitz, Sepehr Nezami, and John Wright for insightful discussions on highest weight orbit learning, random purifications and tomography.
We are grateful to Francesco Anna Mele, Filippo Girardi, Senrui Chen, Marco Fanizza, and Ludovico Lami for sharing a preliminary version of their work~\cite{relatedwork}.
Part of this work was conducted while MW was visiting Q-FARM and the Leinweber Institute for Theoretical Physics at Stanford University and the Simons Institute for the Theory of Computing at UC Berkeley.
This work is supported by the European Union (ERC Grants SYMOPTIC, 101040907 and ASC-Q, 101040624), by the German Federal Ministry of Research, Technology and Space (QuSol, 13N17173), and by the Deutsche Forschungsgemeinschaft (DFG, German Research Foundation, 556164098).

\bibliographystyle{alpha}
\bibliography{library}

\end{document}